\documentclass[11pt]{article}

\usepackage{amsmath,amssymb,amsfonts,amsthm}
\usepackage{fullpage,subfigure,graphicx}
\newcommand{\sam}{n}

\newcommand{\mul}[1]{{n(#1)}}
\newcommand{\prev}[1]{{\varPhi({#1})}}
\newcommand{\cpmx}[1]{{S(#1)}}
\newcommand{\cpm}{{S}}
\newcommand{\estcpm}{{\hat{S}}}
\newcommand{\estcpmx}[1]{{\hat{S}(#1)}}
\newcommand{\dist}{p}
\newcommand{\distper}{p_{\sigma}}
\newcommand{\ps}[1]{{\dist(#1)}}
\newcommand{\perm}[1]{\sigma(#1)}
\newcommand{\sym}{x}
\newcommand{\seq}{\sym^{\sam}}
\newcommand{\est}{{\dist}_\sam}
\newcommand{\optest}{{\est''}_{\sym^{\sam}}}
\newcommand{\optestx}[1]{{\est''}_{\sym^{\sam}}(#1)}
\newcommand{\optestxp}[1]{{\est''}_{\sigma(\sym^{\sam})}(#1)}

\newcommand{\estxs}[1]{{\est}_{\sym^{\sam}}(#1)}
\renewcommand{\estxs}[1]{q_{\sym^{\sam}}(#1)}

\newcommand{\ests}[1]{\est({#1})}

\newcommand{\KL}[2]{D(#1||#2)}

\newcommand{\loss}[2]{L(#1,#2)}

\newcommand{\ed}{\stackrel{\text{def}}{=}}
\newcommand{\iidk}{\Delta_k}

\newcommand{\EE}{\mathbb{E}}

\newcommand{\cX}{\mathcal{X}}

\newcommand{\tcO}{\tilde{\mathcal{O}}}

\newcommand{\cP}{\mathcal{P}}
\newcommand{\cQ}{\mathcal{Q}}
\newcommand{\partit}{\mathbb{P}}

\newcommand{\iid}{\text{i.i.d.}}
\newcommand{\ie}{\text{i.e., }}

\newtheorem{Theorem}{Theorem}

\newtheorem{Lemma}[Theorem]{Lemma}
\newtheorem{Example}[Theorem]{Example}
\newcommand{\eg}{\text{e.g., }}

\newcommand{\ignore}[1]{}
\newcommand{\Xtn}{X^n}

\newcommand{\xts}{x^*}
\newcommand{\Sym}{X}

\newcommand{\Deltak}{\Delta_k}

\newcommand{\upto}{,\ldots,}
\newcommand{\Sets}[1]{\left\{#1\right\}}
\newcommand{\sets}[1]{\{#1\}}
\newcommand{\Paren}[1]{\left(#1\right)}
\newcommand{\bE}{\EE}
\newcommand{\bEXnp}{\operatorname*{\bE}_{X^n\sim p^n}\;}
\newcommand{\PP}{\mathbb{P}}
\newcommand{\cPkc}{\cP_k^{\text{const}}}

\renewcommand{\est}{q}

\renewcommand{\loss}{r}
\newcommand{\rns}[1]{\loss_n(#1)}

\newcommand{\rnqp}{\rns{q,p}}
\newcommand{\rnqP}{\rns{q,P}}
\newcommand{\rnqcP}{\rns{q,\cP}}
\newcommand{\rnqPp}{\rns{q,P'}}
\newcommand{\rncPp}{\rns{q,\cP'}}
\newcommand{\rncP}{\rns{\cP}}
\newcommand{\rnDk}{\rns{\Deltak}}

\newcommand{\rnss}[2]{\loss_n^{#1}(#2)}
\newcommand{\rnPPs}[1]{\rnss{\PP}{#1}}
\newcommand{\rnPPqcP}{\rnPPs{q,\cP}}
\newcommand{\rnPPpqcP}{\rnss{\PP'}{q,\cP}}
\newcommand{\rnP}{\rns{P}}
\newcommand{\rnPp}{\rns{P'}}
\newcommand{\rnPPcP}{\rnPPs{\cP}}
\newcommand{\rnPPpcP}{\rnss{\PP'}{\cP}}

\newcommand{\natest}{\cQ^{\text{nat}}}

\newcommand{\PPprm}{\PP_\sigma}

\newcommand{\rnNs}[1]{\loss_n^{\scriptscriptstyle\text{nat}}(#1)}
\newcommand{\rnNp}{\rnNs{\dist}}
\newcommand{\rnNss}[2]{\loss_n^{\scriptscriptstyle\text{nat}}(#1,#2)}
\newcommand{\rnNDk}{\rnNs{\Delta_k}}

\usepackage{natbib}
\usepackage[backref,colorlinks,citecolor=blue,bookmarks=true]{hyperref}
\usepackage{aliascnt}
\usepackage[numbered]{bookmark}
\begin{document}
\title{Competitive Distribution Estimation}
\author{
\begin{tabular}[t]{c@{\extracolsep{10em}}c} 
 Alon Orlitsky & Ananda Theertha Suresh\\
\small\texttt{alon@ucsd.edu} & \small\texttt{asuresh@ucsd.edu} 
\end{tabular}
\vspace{2ex}\\
University of California, San Diego
}
\maketitle
\begin{abstract}
Estimating an unknown distribution from its samples
is a fundamental problem in statistics.
The common, min-max, formulation of this goal considers the 
performance of the best estimator over all distributions
in a class. It shows that with $n$ samples, distributions
over $k$ symbols can be learned to a KL divergence that
decreases to zero with the sample size $n$, but grows
unboundedly with the alphabet size $k$.

Min-max performance can be viewed as regret relative to an
oracle that knows the underlying distribution. 
We consider two natural and modest limits on the oracle's power.
One where it knows the underlying distribution only up to 
symbol permutations, and the other where it knows the exact
distribution but is restricted to use natural estimators that assign
the same probability to symbols that appeared equally many
times in the sample.

We show that in both cases the competitive regret reduces to
$\min(k/n,\tcO(1/\sqrt n))$, a quantity upper bounded uniformly
for every alphabet size. This shows that distributions can be
estimated nearly as well as when they are essentially known in
advance, and nearly as well as when they are completely known
in advance but need to be estimated via a natural estimator.
We also provide an estimator that runs in linear time and incurs
competitive regret of $\tcO(\min(k/n,1/\sqrt n))$, and show that
for natural estimators this competitive regret is inevitable.
We also demonstrate the effectiveness of competitive estimators using simulations.
\end{abstract}
\section{Introduction}
\label{sec:introduction}
\subsection{Background}
The basic problem of learning an unknown distribution from its
samples is typically formulated in terms of \emph{min-max} 
performance with KL-divergence loss.
Though simple and intuitive to state, 
its precise formalization requires a modicum of nomenclature. 

Let $\cP$ be a known collection of distributions over a discrete set
$\cX$, and let $X_1,X_2,\ldots$ be samples generated independently
according to $\dist$.
%We would like to estimate an unknown distribution $p\in\cP$ from 
%independent samples it generates.
A distribution \emph{estimator}
%over $\cX$ is a mapping $\est:\cX^*\to\cP$ that
$\est$ over $\cX$ 
associates with any observed sample sequence
$\xts\in\cX^*$ a distribution $\est_{\xts}$ over $\cX$.
The performance of $\est$ is evaluated in terms of a given distance measure,
and we will use the popular \emph{KL divergence}
\[
\KL{\dist}{\est}
\ed
\sum_{\sym \in \cX} \ps{\sym} \log \frac{\ps{\sym}}{\ests{\sym}}.
\]
KL divergence reflects the increase in the number
of bits over the entropy needed to compress the output of $p$
using an encoding based on $q$ and the log-loss of estimating
$p$ by $q$, see~\cite{CoverT06}.
%From here on, all logarithms will be natural.

Given $n$ samples $\Xtn\ed\Sym_1,\Sym_2\upto\Sym_{\sam}$
generated independently according to $\dist$, 
the expected loss of the estimator $\est$ is 
\[
\rnqp
=
\bEXnp[\KL{\dist}{\est_{{}_{\Xtn}}}],
\]
the worst-case loss of $\est$ for any distribution in $\cP$ is 
\begin{equation}
\label{eqn:loss_q_P}
\rnqcP
\ed
\max_{\dist\in\cP}\rnqp,
\end{equation}
and the lowest worst-case loss for $\cP$, achieved by the best estimator is
\begin{equation}
\label{eqn:loss_P}
\rncP
\ed
\min_{\est} \rnqcP.
\end{equation}
Namely, $\rncP$ is the min-max loss
\[
\rncP
=
\min_{\est}\max_{\dist\in\cP}\rnqp.
\]
Min-max performance can be viewed as regret relative to an
oracle that knows the underlying distribution. 
Thus we refer to the above quantity as regret from here on.
%Unlike the Bayesian apporach, where the estimator is known for every
%class $\cP$, the frequentist results are known only for few classes
%$\cP$,

\subsection{Prior work}
The most natural and important collection of distributions
is the set of all distributions over the alphabet $\cX$.
To simplify notation, assume without loss of generality
that the alphabet is $[k]=\{1,2,\ldots k\}$, and then the
set of all distributions is the \emph{simplex} in $k$ dimensions,
\[
\Delta_k
\ed
\Sets{(p(1)\upto p(k))\ :\ 
\forall 1\le i\le k,  \,\,
%\operatorname*{\forall}_{i=1}^k 
p(i)\ge 0\text{ and }\sum_{i=1}^k p(i)=1}.
\]
\cite{Krichevsky98} introduced the problem of estimating
$\rns{\Deltak}$, and~\cite{BraessT04} showed that as $k/n\to0$,
\begin{equation}
\label{eqn:km1over2}
\rns{\Deltak} = \frac{k-1}{2n} + o\Paren{\frac kn}.
\end{equation}
This result specifies the rate at which distributions in $\Deltak$
can be approximated in KL divergence as the number of samples increases.
It also implies the upper bound of the well known 
$\frac{k-1}2\log n$ redundancy of \iid\ distributions,
derived by~\cite{KrichevskyT81}.
Other loss measures, including $\ell_1$, were considered
in~\cite{KamathOPS15}, and related results appeared in~\cite{HanJW14}.

Motivated by natural-language processing, bioinformatics, and other modern
applications, there has been a fair amount of recent interest in
evaluating and achieving the optimal regret in the non-asymptotic
regime, and when the sample size is not overwhelmingly larger than the
alphabet size. For example in English text processing, the alphabet
is English vocabulary whose size is comparable to the 
number of times a context has appeared in the corpus.

It can be shown that when the sample size $n$ is linear in 
the alphabet size $k$, $\rnDk$ is a constant,
and~\cite{Paninski04} showed that as $k/n\to\infty$, 
\[
\rnDk
=
\log\frac kn + o\Paren{\log\frac kn}.
\]
If follows that distributions cannot be well learned with a number
of samples that is comparable to their alphabet size.

Several modifications have been proposed to address this problem.
\cite{OrlitskySZ03} modified the loss to reflect
estimating the probabilities of each previously-observed symbol,
and all unseen symbols combined. They showed that the corresponding
regret can be upper bounded in terms of the number of samples $n$,
regardless of the alphabet size $k$.
\cite{McAllesterS00, DrukhM04, AcharyaJOS13} estimated
the combined probability of symbols that appeared a given number of times.
Several others have restricted the collections of distributions
in $\iidk$ to monotone, unimodal, or log-concave distributions, see~\eg~\cite{Birge87, ChanDSS13}.
%\cite{BoucheronGO13} restricted the collection of distributions in
%$\Deltak$ to \emph{envelope distributions}. They showed that ...
%Other restrictions on the distribution class were addressed for
%$\ell_1$ loss~\cite{}.

In this paper we address the original problem, with KL divergence
regret, and the loss for the whole collection $\Deltak$.
However, instead of considering min-max regret we take a
competitive approach where we compare 
and show that it is possible to
learn the distribution with a uniformly-bounded regret.

\section{Competitive formulation}
\subsection{Background}
While~\eqref{eqn:km1over2} is asymptotically tight, it addresses the
worst-case regret over all possible distributions in $\Deltak$. 
For smaller distribution collections in $\Deltak$ lower regret may
be achieved.

Our goal is to derive a data-driven estimator that approaches the
performance of the best estimator for any reasonable sub-collection
of $\Deltak$.

\begin{Example}
Consider the constant-$i$ distribution over $[k]$,
\[
p_i(j)
=
\begin{cases}
1 & \text{ for } j=i,\\
0 & \text{ for } j\ne i,
\end{cases}
\]
and the collection of all constant distributions,
\[
\cPkc
\ed
\Sets{p_1\upto p_k}
\subseteq
\Deltak.
\] 
Clearly, for all $n\ge1$,
\[
\rns{\cPkc}
=
0
\]
as the data-driven estimator $\est$ that assigns probability $1$ to the seen
symbol $X_1$ and probability $0$ to all other $k-1$ symbols, estimates
$\dist$ exactly.
\end{Example}

Our goal is to derive a single estimator that simultaneously achieves
essentially the lowest regret possible for every reasonable
sub-collection of $\Deltak$.

Note that the definition of loss (regret) can be viewed as competing
with a person who knows the underlying distribution $\dist$ and can
use any estimator $\est$, naturally choosing $\dist$ itself. 
%Hence, for the rest of the paper we refer to both loss and regret as regret.

Two simple relaxations of the problem clearly come to mind.
First, competing with a person who has only partial information
about $\dist$. And second, competing with a person who knows
$\dist$ but can only use a restricted type of estimator, in
particular, only  estimators that would arise naturally.

We define these two modified regrets and show that under both the modified regrets, one can uniformly bound the regret.

\subsection{Competing with partial information}
One way to weaken an oracle-based estimator is to provide it with less
information. Consider an oracle who, instead of knowing $\dist\in\cP$ exactly,
has only partial knowledge of $\dist$. For simplicity, we interpret
partial knowledge as knowing the value of $f(\dist)$ for a given function $f$
over $\cP$. 
Any such $f$ partitions $\cP$ into subsets, each corresponding
to one possible value of $f$, and henceforth, we will use
this equivalent formulation, namely $\PP$ is a known partition 
of $\cP$, and the oracle knows the unique partition part $P$
such that $\dist\in P\in\PP$. 

For every partition part $P\in\PP$, an estimator $\est$ incurs
the worst-case regret in~\eqref{eqn:loss_q_P},
\[
\rnqP
=
\max_{p\in P}\rnqp.
\]
The oracle, knowing $P$, incurs the least worst-case regret~\eqref{eqn:loss_P},
\[
\rnP
=
\min_q\rnqP.
\]
The \emph{competitive regret} of $\est$ over the oracle, for all distributions
in $P$ is
\[
\rnqP-\rnP,
\]
the competitive regret over all partition parts and all
distributions in each is 
\[
\rnPPqcP
\ed
\max_{P\in\PP}
\Paren{\rnqP-\rnP},
\]
and the best possible competitive regret is
\[
\rnPPcP
\ed
\min_{\est} \rnPPqcP.
\]
Consolidating the intermediate definitions,
\[
\rnPPcP
=
\min_{\est}
\max_{P\in\PP}
\Paren{\max_{p\in P}\rnqp-\rnP}.
\]
Namely, an oracle-aided estimator who knows the partition part
incurs a worst-case regret $\rnP$ over each part $P$,
and the competitive regret $\rnPPcP$ of data-driven estimators
is the least overall increase in the part-wise regret due
to not knowing $P$.
The following examples evaluate $\rnPPcP$ for the two simplest partitions
of any collection $\cP$.
\begin{Example}
The \emph{singleton partition} consists of $|\cP|$ parts,
each a single distribution in $\cP$,
\[
\PP_{|\cP|}
\ed
\Sets{\sets{p}:p\in\cP}.
\]
An oracle-aided estimator that knows the part containing $\dist$ knows $\dist$.
The competitive regret of data-driven estimators is therefore the min-max regret,
\begin{align*}
\rnss{\PP_{|\cP|}}{\cP}
&=
\min_{\est}
\max_{p\in\cP}
%\Paren{\max_{p'\in \sets{p}}\rns{q,p'}-\rns{\sets p}}\\
\Paren{\rns{\est,\sets p}-\rns{\sets p}}\\
&=
\min_{\est}
\max_{p\in\cP}
\rnqp\\
&=
\rncP,
\end{align*}
where the middle equality follows as $\rns{\est,\sets p}=\rnqp$,
and $\rns{\sets p}=0$.
\end{Example}

\begin{Example}
The \emph{whole-collection} partition has only one part,
the whole collection $\cP$,
%The two simplest partitions of a collection $\cP$ are
\[
\PP_1
\ed
\Sets{\cP}.
\]
An estimator aided by an oracle that knows the part containing $\dist$
has no additional information, hence no advantage over a data-driven estimator,
and the competitive regret is 0,
\begin{align*}
\rnss{\PP_1}{\cP}
&=
\min_{\est}
\max_{P\in\sets{\cP}}
\Paren{\max_{p\in P}\rnqp-\rnP}\\
&=
\min_{\est}
\Paren{\max_{p\in\cP}\rnqp-\rncP}\\
&=
\min_{\est}
\max_{p\in\cP}
\Paren{\rnqp}-\rncP\\
&=
\rncP-\rncP\\
&=0.
\end{align*}
\end{Example}

The examples show that for the coarsest partition of $\cP$, 
into a single part, the competitive regret is the lowest
possible, 0, while for the finest partition, into singletons,
the competitive regret is the highest possible, $\rncP$.

A partition $\PP'$ \emph{refines} a partition $\PP$
if every part in $\PP$ is partitioned by some parts in $\PP$.
For example $\Sets{\sets{a,b},\sets{c},\sets{d,e}}$
refines $\Sets{\sets{a,b,c},\sets{d,e}}$.
It is easy to see that if $\PP'$ refines $\PP$ then for every $\est$
\begin{equation}
\label{eq:refine}
\rnPPpqcP
\ge
\rnPPqcP
\end{equation}
The definition implies that if $\cP'\subseteq\cP$ then
$\rncPp\le\rncP$, hence for every $\est$,
\begin{align*}
\rnPPpqcP
&=  
\max_{P'\in\PP'}\Paren{\rnqPp - \rnPp}\\
&=
\max_{P\in\PP}\max_{P\supseteq P'\in\PP'}\Paren{\rnqPp - \rnPp}\\
&\ge
\max_{P\in\PP}\max_{P\supseteq P'\in\PP'}\Paren{\rnqPp - \rnP}\\
&=
\max_{P\in\PP}\bigl(\max_{P\supseteq P'\in\PP'}\rnqPp - \rnP\bigr)\\
&=
\max_{P\in\PP}\Paren{\rnqP - \rnP}\\
&=
\rnPPqcP.
\end{align*}

Note that this notion of competitiveness has appeared in several contexts.
In data compression it is called
\emph{twice-redundancy}~\cite{Ryabko84,
  Ryabko90,BontempsBG14,BoucheronGO14}, 
while in statistics it often called \emph{adaptive}
or \emph{local min-max} see \eg \cite{DonohoJ94, AbramovichBDJ06, BickelKRW93, BarronBM99, Tsybakov04}, 
and recently in property testing it is referred as competitive~\cite{AcharyaDJOP11,AcharyaDJOPS12,AcharyaJOS13b} or \emph{instance-by-instance}~\cite{ValiantV14}.

\subsection*{Permutation class}

Considering the collection $\Deltak$ of all distributions over $[k]$,
if follows that as we start with single-part partition $\Sets{\Deltak}$
and keep refining it till the oracle knows $p$, the competitive regret
of estimators will increase from 0 to $\rns{\Deltak}$.

A natural question is therefore how much information can the oracle
have and still keep the competitive regret low. We show that the
oracle can know the distribution exactly up to permutation, and
still the relative regret will be very small.

Call two distributions $p$ and $p'$ over $[k]$ \emph{permutation equivalent} if
there is a permutation $\sigma$ of $[k]$ such that
\[
p'_{\sigma(i)}
=
p_i,
\]
for example, over $[3]$, $(0.5, 0.3, 0.2)$ and $(0.3, 0.5, 0.2)$
are permutation equivalent. Permutation equivalence is clearly
an equivalence relation, and hence partitions the collection of
distributions of $[k]$ into equivalence classes. 
Let $\PPprm$ be the corresponding partition.
We construct estimators $q$ that uniformly bound $\rnss{\PPprm}{q, \Delta_k}$,
thus the same estimator uniformly bounds $\rnss{\PP}{q,\Delta_k}$
for any coarser partition of $\Delta_k$ such as partitions such 
that each class contains distributions with same entropy, or same sparse-support.
%
%\[
%\rns{\PPprm},
%\]
%and the estimator that achieves it.

%Furthermore, the
%class of distributions we are consdiering is also called
%\emph{canonical}~\cite{} or \emph{very natural} distributions~\cite{}.

\subsection{Competing with natural estimators}
Another restriction on the oracle-aided estimator is to still
let it know  $\dist$ exactly, but force it to be ``natural'', namely,
to assign the same probability to all symbols that appeared the
same number of times in the sample. For example, for the observed
sample $a,b,c,a,b,d,e$, to assign the same probability
to $a$ and $b$, and the same probability to $c$, $d$, and $e$.

\ignore{
Formally, let $\mul{\sym}$ denote the number of times a symbol $\sym$
appeared in the sample $\Xtn$.
An estimator $\est$ is \emph{natural} if $\estxs{\sym} =
\estxs{\sym'}$ whenever $\mul{\sym} = \mul{\sym'}$.
%We let $\natest$ to denote the set of all natural estimators.
}

Since data-driven estimators derive all their knowledge of the
distribution from the data, we expect them to be natural. We
also saw in the previous section that natural estimators are
optimal under for the permutation-invariant oracle.
%Natural estimators are called~\emph{canonical} in~\cite{BatuFFKRW01}.

We now compare the regret of data-driven estimators to that of
\emph{natural oracle-aided} estimators.
For a distribution $\dist$, the lowest regret of natural estimators is
\[
\rnNp
\ed
\min_{\est \in \natest}\rnqp,
\]
where $\natest$ is the set of all natural estimators.
The regret of an estimator $q$ relative to the best
natural-estimator designed with knowledge of $p$ is 
\[
\rnNss{\est}{\dist}
=
\rnqp - \rnNp.
\]
The regret of data-driven estimators over $\cP$ is therefore,
\[
\rnNs{\cP}
=
\min_{\est} \max_{\dist \in \cP} \rnNss{\est}{\dist}.
\]
We show that indeed for the class of discrete distributions over support $[k]$, \ie $\Delta_k$, $\rnNDk$ is uniformly bounded.
%
%If we want to combine this approach before previous approach, we can again ask the local-optima for every estimator 
%with natural estimators for that we have
%\[
%\nerrorn{\cP} = \min_{\est} \max_{\dist \in \cP}\error{\est}{\dist},
%\]
%and talk about local optima under min-max estimators.

\ignore{
--------
Let $\distper$ for a permutation $\sigma$ is given as follows
\[
\distper(x) = \dist(\perm{x}).
\]
We consider a permutation partition class defined as follows.
$\PP_{\sigma}$ is the finest partition possible such that 
if $\dist\in \cP \in \PP_{\sigma}$, then  for every $\sigma$ should be in the same class $\cP$.
In other words its the partition such that every single distribution within each class has the same multiset.
Such classes are also natural $\ie$ if we expect an estimator to work on distribution $\dist(1)=1, \dist(i)=0 \forall i \neq 1$,
then we expect it to work well for distribution $\dist'(2) = 1, \dist'(i) = 0 \forall i \neq 2$.
}
The rest of the paper is organized as follows.
In Section~\ref{sec:results}, we state our results and in Section~\ref{sec:proofs}, provide the proofs.
In Section~\ref{sec:experiments}, we compare the competitive estimator to other min-max motivated estimators using experiments.

\section{Results}
\label{sec:results}
We show that the estimator proposed in~\cite{AcharyaJOS13} (call it  $q'$)
is competitive if partial information is known or
if we restrict the class of estimators to be natural.
In Theorem~\ref{thm:proposed}, we prove
\begin{equation}
\label{eq:main_result}
%\rnss{ \PP}{\est',\Delta_k} \leq 
\rnss{ \PP_{\sigma}}{\est',\Delta_k} \leq \rnNss{\est'}{\Delta_k} \leq \tcO \left( \min \left( \frac{1}{\sqrt{n}},\frac{k}{n} \right)\right).
\end{equation}
Thus for any coarser partition $\PP$, the same result holds.
%for any coarser partition $\PP$ of $\PPprm$. 
Here $\tcO$ and later $\tilde\Omega$ hide multiplicative logarithmic factors.
\ignore{
Since in Equation~\eqref{eq:refine}, we show that any coarser partition incurs smaller error, we have for any $\PP$ that is a coarser than $\PP_{\sigma}$,
\[
\rnss{\PP}{\est', \Delta_k} \leq  \tcO \left( \min \left( \frac{1}{\sqrt{n}},\frac{k}{n} \right)\right).
\]}
Together with Lemma~\ref{lem:labelbounded}, the lower bounds in~\cite{AcharyaJOS13} can be extended to show that \[
 \rnNDk \geq \tilde\Omega\left( \min \left( \frac{1}{\sqrt{n}},\frac{k}{n} \right)\right).
\]
Thus the performance of the estimator proposed in~\cite{AcharyaJOS13} is nearly-optimal compared to the class of natural estimators.
Equation~\eqref{eq:main_result} immediately implies
$\rnss{\PP_{\sigma}}{\Delta_k} \leq \rnNDk \leq 
\tcO \left( \min \left( \frac{1}{\sqrt{n}},\frac{k}{n} \right)\right)$. However Equation~\eqref{eqn:km1over2} and  the fact that 
the min-max estimator proposed in~\cite{BraessT04} is natural imply
\begin{align*}
\rnNDk & = \min_{\est} \max_{p \in \iidk} \left( \EE[\KL{p}{\est}] - \rnNp \right) \\
& \leq \min_{\est} \max_{p \in \iidk} \EE[\KL{p}{\est}] \\
& = \frac{(k-1)(1+o(1))}{2n}.
\end{align*}
The above equation together with 
Equation~\eqref{eq:main_result}
implies a stronger bound:
\[ 
\rnss{\PP_{\sigma}}{\Delta_k} \leq \rnNDk \leq \min \left( \tcO \left( \frac{1}{\sqrt{n}} \right), \frac{(k-1)(1+o(1))}{2n} \right).
\]

\ignore{
\section{Basic properties}
We first show that if we have an upper bound on $\rnPPpqcP$,
then the same upper bound holds for every finer partition of $\PP'$.
As a corollary, it implies that our results hold for every coarse
partition of $\PP_{\sigma}$.
\begin{Lemma}
\label{lem:finer}
If $\PP'$ refines $\PP$, then for every estimator $\est$,
\[
\rnPPqcP \le \rnPPpqcP,
\]
hence, 
\[
\rnPPcP \leq \rnPPpcP.
\]
\end{Lemma}
\begin{proof}
Observe that $\rnP$ increases, namely,
if $\cup_{i} P_i = P$ and for all $i,j$, $P_i \cap P_j = \emptyset$, then
\begin{equation}
\label{eq:union}
\max_{i\in \{1,2\}} \rns{P_i} \le \rnP.
\end{equation}
Let $\PP=\{P_1,P_2\ldots\}$ and $\PP' = \{P'_1,P'_2,\ldots
\}$,
then
\begin{align*}
  \rnPPqcP
%\rnPPqcP
& =  
\max_{P \in \PP} \left( \max_{\dist \in P} \EE[\KL{\dist}{\est}] - \rnP \right) \\
& = 
\max_{P'\in \PP'} \max_{P: P \subseteq P'} \left( \max_{\dist \in P} \EE[\KL{\dist}{\est}] - \rnP \right) \\
& \stackrel{(a)}{\geq }
\max_{P'\in \PP'} \left( \max_{P: P \subseteq P'}  \max_{\dist \in P} \EE[\KL{\dist}{\est}]  -\max_{P: P \subseteq P'}  \rnP \right) \\
&  \stackrel{(b)}{\geq  }
\max_{P'\in \PP'} \left(  \max_{\dist \in P'} \EE[\KL{\dist}{\est}]  - \rns{P'} \right)\\
& = \rnPPpqcP.
\end{align*}
$(a)$ follows from the fact that maximum of differences is bigger than difference of maximums. Equation~\ref{eq:union} results in $(b)$.
\end{proof}
}
\section{Proofs}
\label{sec:proofs}
The proof consists of two parts. We first show that for every estimator $q$, $ \rnss{\PPprm}{q,\iidk} \leq \rnNss{\est}{\iidk}$ and then upper bound $\rnNs{\est,\iidk}$ using results on combined probability mass.
\subsection{Relation between $ \rnss{\PPprm}{q,\iidk}$ and $\rnNs{\est,\iidk}$}
We now show an auxiliary result that helps us relate $\rnNss{\est}{\iidk}$ to $\rnss{\PPprm}{q,\iidk}$. For a symbol $x$, let $\mul{x}$ denote the number of times it appears in the sequence.
\begin{Lemma}
\label{lem:technical}
For every class $P \in \partit_{\sigma}$, 
\[
\rnP \geq \max_{\dist \in P} \rnNs{\dist}.
\]
\end{Lemma}
\begin{proof}
We first show that the estimator that there is an optimal estimator $\est$ that is natural. In particular,
let 
\[
\optestx{y} = 
\frac{\sum_{\sigma \in \Sigma_k}  \dist(\sigma(\seq y))} {\sum_{\sigma' \in \Sigma_k} \dist(\sigma'(\seq))},
\]
where $\Sigma_k$ is the set of all permutations of $k$ symbols.
We show that $\optestx{y}$ is an optimal estimator for $P$.
Since $\optestx{y} = \optestxp{\sigma(y)}$ for any permutation $\sigma$, the estimator achieves the same loss for every $\dist \in P$
\begin{equation}
\label{eq:sameforall}
 \max_{\dist \in P} \rns{\optest,\dist} 
= \frac{1}{k!} \sum_{\sigma \in \Sigma_k} \rns{\optest,\dist(\sigma(\cdot))}.
\end{equation}
For any estimator $q$,
\begin{align*}
%\rnP 
\max_{p \in P} \EE[\KL{\dist}{\est}] 
& \stackrel{(a)}{\geq} \frac{1}{k!} \sum_{\sigma \in \Sigma_k}  \EE_{\dist(\sigma(\cdot))}[\KL{\dist(\sigma(\cdot))}{\est}] \\
& \stackrel{(b)}{=}  \frac{1}{k!} \sum_{\sigma \in \Sigma_k}
\sum_{\seq \in \cX^n} \sum_{y \in \cX} p(\sigma(x^ny)) \log \frac{1}{\estxs{y}} - H(p) \\
& =  \frac{1}{k!} \sum_{\seq \in \cX^n}
    \sum_{\sigma \in \Sigma_k}
   \sum_{y \in \cX} p(\sigma(x^ny)) \log \frac{1}{\estxs{y}}  - H(p) \\
& \stackrel{(c)}{\geq}  \frac{1}{k!} \sum_{\seq \in \cX^n}
    \sum_{\sigma \in \Sigma_k}
    \sum_{y \in \cX}   p(\sigma(x^ny)) \log \frac{  \sum_{\sigma' \in \Sigma_k}
   p(\sigma'(x^n))}{  \sum_{\sigma'' \in \Sigma_k}
 p(\sigma''(x^ny))} - H(p)\\
%& \stackrel{(c)}{=}  \frac{1}{k!} \sum_{\seq \in \cX^n}
%  (\sum_{\sigma' \in \Sigma_k} \dist(\sigma'(\seq \in \cX^n))) \sum_{y}
% \frac{  \sum_{\sigma \in \Sigma_k}  p(\sigma(x^ny)) }{ (\sum_{\sig%ma'} \dist(\sigma'(\seq \in \cX^n)))}
%\log \frac{1}{\estxs{y}}  - H(p) \\
%& \stackrel{(d)}{\geq} \frac{1}{k!} \sum_{\seq \in \cX^n}   (\sum_{\sigma'} \dist(\sigma'(\seq \in \cX^n))) 
% \sum_{y}  \frac{  \sum_{\sigma}  p(\sigma(x^ny)) }{ (\sum_{\sigma'%} \dist(\sigma'(\seq \in \cX^n)))} \log \frac{1}{\optestx{y}} - H(p) \\
& =  \frac{1}{k!} \sum_{\sigma \in \Sigma_k}  \sum_{\seq \in \cX^n} \sum_{y \in \cX} \dist{(\sigma(\seq y))} \log \frac{1}{\optestx{y}} - H(p) \\
& \stackrel{(d)}{=}  \frac{1}{k!} \sum_{\sigma \in \Sigma_k} \rns{\optest,\dist(\sigma(\cdot))}.
%& \geq \error{\optest}{\dist}.
\end{align*}
$(a)$ follows from the fact that maximum is larger than the average.
$(b)$ follows from the fact that every distribution in $P$ has the same entropy. Non-negativity of KL divergence implies $(c)$. 
All distributions in $P$ has the same entropy and hence $(d)$.
Hence together with Equation~\eqref{eq:sameforall}
\begin{align*}
\rnP & = \min_{\est} \max_{p \in P} \EE[\KL{p}{q}] \\
& \geq \frac{1}{k!} \sum_{\sigma \in \Sigma_k} \rns{\optest,\dist(\sigma(\cdot))} \\
& = \max_{p \in P} \rns{\optest, p}.
\end{align*}
Hence $\optest$ is an optimal estimator. $\optest$ is natural as if $n(y) = n(y')$, then $\optestx{y} = \optestx{y'}$. Since there is a natural estimator that achieves minimum in $\rnP$,
\begin{align*}
\rnP 
& = \min_{\est} \max_{\dist \in P} \EE[\KL{\dist}{\est}] \\
& = \min_{\est \in \natest} \max_{\dist \in P} \EE[\KL{\dist}{\est}] \\
& \geq\max_{\dist \in P}  \min_{\est \in \natest}  \EE[\KL{\dist}{\est}] \\
& = \max_{\dist \in P} \rnNp,
\end{align*}
where the last inequality follows from the fact that min-max is bigger than max-min.
\end{proof}
We now relate $\rnNss{\est}{\iidk}$ to $\rnss{\PPprm}{q,\iidk}$.
\begin{Lemma}
\label{lem:onelessthanother}
For every estimator $\est$,
\[
\rnss{\PPprm}{q,\iidk} \leq \rnNss{\est}{\iidk}.
\]
\end{Lemma}
\begin{proof}
\begin{align*}
\rnss{\PPprm}{q,\iidk} 
& = \max_{P \in \partit_{\sigma}} \left(\max_{\dist\in P}\EE[\KL{\dist}{\est}] - \rnP \right) \\
& \stackrel{(a)}{\leq}  \max_{P \in \partit_{\sigma}} \left(\max_{\dist\in P}\EE[\KL{\dist}{\est}] - \max_{\dist \in P} \rnNp \right) \\
& \stackrel{(b)}{\leq}  \max_{P \in \partit_{\sigma}} \max_{\dist\in P} \left(\EE[\KL{\dist}{\est}] - \rnNp \right) \\
& = \max_{\dist \in \iidk}  \left(\EE[\KL{\dist}{\est}] - \rnNp \right) \\
& = \rnNs{\est,\iidk}.
\end{align*}
$(a)$ follows from Lemma~\ref{lem:technical}.
Difference of maximums is smaller than maximum of differences, hence $(b)$.
\end{proof}
\subsection{Relation between $\rnNs{q,\iidk}$ and combined probability estimation}

We now relate the regret in estimating distribution to that of estimating the combined or total probability mass, defined as follows.
For a sequence $\seq$, let $\prev{t}$ denote the number of symbols appearing $t$ times   and $\cpmx{t}$ denote the total probability of symbols appearing $t$ times.
Similar to KL divergence between distributions, we define KL divergence between $\cpm$ and their estimates $\estcpm$ as 
\[
\KL{\cpm}{\estcpm} = \sum^n_{t=0} \cpmx{t} \log \frac{\cpmx{t}}{\estcpmx{t}}.
\]
We find the best natural estimator that minimizes $\rnNp$ in the next lemma.
\begin{Lemma}
\label{lem:bestnatural}
Let $q^*(x) = \frac{\cpmx{n(x)}}{\prev{n(x)}}$, then
\[
q^* = \arg \min_{q \in \natest} \rnqp
\]
and 
\[
\rnNp =  \EE  \left[\sum^{n}_{t=0} \cpmx{t} \log  \frac{\prev{t}}{\cpmx{t}} \right] - H(p).
\]
\end{Lemma}
\begin{proof}
For a sequence $\seq$ and estimator $q$,
\begin{align*}
\sum_{\sym \in \cX} \ps{\sym} \log \frac{1}{\estxs{\sym}}
- 
\sum^{n}_{t=0} \cpmx{t} \log \frac{\prev{t}}{\cpmx{t}} 
& = \sum^{n}_{t=0} \sum_{\sym: \mul{\sym}=t} 
\ps{\sym} \log \frac{1}{\estxs{\sym}} 
-
\sum^{n}_{t=0} \cpmx{t} \log \frac{\prev{t}}{\cpmx{t}}\\
& = \sum^{n}_{t=0} \cpmx{t} 
\log \frac{1}{\estxs{\sym}} 
-
\sum^{n}_{t=0} \cpmx{t} \log \frac{\prev{t}}{\cpmx{t}}\\
& = \sum^{n}_{t=0} \cpmx{t} \log \frac{\cpmx{t}}{\prev{t}\estxs{\sym}} \\
& \geq 0,%\sum^{n}_{t=0} \cpmx{t} \log  \frac{\prev{t}}{\cpmx{t}},
\end{align*}
where the last inequality follows from the fact that $\sum_{t} \cpmx{t} = \sum_{t} \prev{t} \estcpmx{t} = 1$ and KL divergence is non-negative. Furthermore the last inequality is achieved only by the estimator that assigns $\est^*(x) = \frac{\cpmx{n(x)}}{\prev{n(x)}}$.
Hence,
\[
\rnNp = \min_{\est \in \natest} \EE \left[\sum_{\sym \in \cX} \ps{\sym} \log \frac{\ps{\sym}}{\estxs{\sym}}
\right]
=  -H(p) + \EE  \left[\sum^{n}_{t=0} \cpmx{t} \log  \frac{\prev{t}}{\cpmx{t}} \right] .
\]

\end{proof}
Since the natural estimator assigns same probability to symbols that appear the same number of times, estimating probabilities is same as estimating the total probability of symbols appearing a given number of times. We formalize it in the next lemma.
\begin{Lemma}
%\label{lem:labelbounded}
For a natural estimator $\est$
%
%$\mul{x} = t$, then
let  $\estcpmx{t} = \sum_{x: \mul{x} = t} \ests{x}$, then
\[
\rns{\est,\dist} = \EE[\KL{\cpm}{\estcpm}].
\]
\end{Lemma}
\begin{proof}
For any estimator $\est$ and sequence $\seq$,
\begin{align*}
\sum_{\sym \in cX} \ps{\sym} \log \frac{1}{\estxs{\sym}}
& = \sum^{n}_{t=0} \sum_{\sym: \mul{\sym}=t} 
\ps{\sym} \log \frac{1}{\estxs{\sym}} \\
& = \sum^{n}_{t=0} \cpmx{t} \log \frac{\cpmx{t}}{\prev{t}\estxs{\sym}} +
\sum^{n}_{t=0} \cpmx{t} \log \frac{\prev{t}}{\cpmx{t}} \\
& =  \sum^{n}_{t=0}  \cpmx{t} \log \frac{\cpmx{t}}{\estcpmx{t}} +
\sum^{n}_{t=0} \cpmx{t} \log \frac{\prev{t}}{\cpmx{t}}.
\end{align*}
Thus by Lemma~\ref{lem:bestnatural},
\begin{align*}
\rns{\est,\dist} 
& =    -H(p) + \EE\left[  \sum^n_{t=0}\cpmx{t} \log \frac{\cpmx{t}}{\estcpmx{t}} +
\sum^{n}_{t=0} \cpmx{t} \log \frac{\prev{t}}{\cpmx{t}}  \right] 
 + H(p) -  \EE  \left[\sum^{n}_{t=0} \cpmx{t} \log  \frac{\prev{t}}{\cpmx{t}} \right] \\
& =  \EE\left[ \sum^{n}_{t=0}  \cpmx{t} \log \frac{\cpmx{t}}{\estcpmx{t}} \right] \\
& =   \EE[\KL{\cpm}{\estcpm}].
\end{align*}
%Hence,
%\[
%\rnNDk
%= 
%\min_{\est} \rnNss{\est}{\iidk} = \min_{\estcpm} \max_{\dist \in \iidk} \EE[\KL{\cpm}{\estcpm}].
%\]
\end{proof}
Taking maximum over all distributions $\dist$ and minimum over all estimators $q$ results in
\begin{Lemma}
\label{lem:labelbounded}
For a natural estimator $\est$
%
%$\mul{x} = t$, then
let  $\estcpmx{t} = \sum_{x: \mul{x} = t} \ests{x}$, then
\[
\rnNss{\est}{\iidk} = \max_{\dist \in \iidk} \EE[\KL{\cpm}{\estcpm}].
\]
Furthermore,
\[
\rnNDk =   \min_{\estcpm} \max_{\dist \in \iidk} \EE[\KL{\cpm}{\estcpm}].
\]
\end{Lemma}
Thus finding the best competitive natural estimator is same as finding the best estimator for the combined probability mass $\cpm$. 
\cite{AcharyaJOS13} proposed an algorithm for estimating $\cpm$
such that for all $k$ with probability $\geq 1-1/n$,
\[
\max_{\dist \in \iidk} [\KL{\cpm}{\estcpm} = \tcO \left(\frac{1}{\sqrt{n}} \right).
\]
The result is stated in Theorem $2$ of~\cite{AcharyaJOS13}.
One can convert this result to a result on expectation easily using the property that their estimator is bounded below by $1/2n$ and show that 
\[
\max_{\dist \in \iidk} \EE[\KL{\cpm}{\estcpm}] = \tcO \left(\frac{1}{\sqrt{n}} \right).
\]
A slight modification of their proof for Lemma ${17}$ and Theorem $2$ in their paper using $\sum^n_{t=1} \sqrt{\prev{t}} \leq \sum^n_{t=1} \prev{t} \leq k$ shows that their estimator $\hat{S}$ for the combined probability mass $S$ satisfies
\[
 \max_{\dist \in \iidk} \EE[\KL{\cpm}{\estcpm}]
= \tcO\left(\min \left( \frac{1}{\sqrt{n}},\frac{k}{n} \right) \right).
\]
Above equation together with Lemmas~\ref{lem:onelessthanother} and~\ref{lem:labelbounded} shows that
\begin{Theorem}
\label{thm:proposed}
For any $k$ and $n$,  the proposed estimator $\est$ in~\cite{AcharyaJOS13} satisfies
\[
\rnss{\PPprm}{q,\iidk} \leq \rnNss{\est}{\iidk} \leq \tcO\left(\min \left( \frac{1}{\sqrt{n}},\frac{k}{n} \right) \right).
\]
\end{Theorem}

\section{Experiments}
\label{sec:experiments}
For small values of $n$ and $k$ 
the estimator proposed in~\cite{AcharyaJOS13} 
 behaves as a combination of Good-Turing and empirical estimator.
Hence for experiments, we use the following combination of Good-turing and empirical estimators
%For a symbol appearing $t$ times, 
%\[
%q_C(x) = \frac{C(t)}{N\prev{t}},
%\]
%where $N$ is the normalization factor to ensure the probabilities add up to $1$ and 
\begin{equation*}
q(x)= \begin{cases}
        \frac{\mul{x}}{N} &\text{ if } \mul{x} > \prev{\mul{x}+1},
        \\
        \frac{\max(\prev{\mul{x}+1},1)}{\prev{\mul{x}}} \cdot \frac{\mul{x}+1}{N} & \text{ else,}
\end{cases} 
\end{equation*}
where $N$ is the normalization factor to ensure that the probabilities add to $1$.
%The above estimator $C$ is known to achieve 
%\[
%\rnNDk \leq \tcO \left(\min \left(\frac{1}{n^{1/3}}, \frac{k}{n} \right) \right).
%\]
We compare the above competitive estimator with several popular add-$\beta$ estimators $\hat{S}$ of the form
\[
q_{\hat{S}}(x) = \frac{t+\beta_{\hat{S}}(t)}{N(\hat{S})},
\]
where $N(\hat{S})$ is a normalization factor to ensure that the probabilities add up to $1$.

The Laplace estimator has $\beta_{L}(t) = 1 \, \forall \, t$. It is optimal when the underlying distribution is generated from the uniform prior on  $\iidk$. The  Krichevsky-Trofimov estimator 
 has $\beta_{KT}(t) = 1/2 \, \forall \, t$ and is min-max optimal for the cumulative regret or when the underlying distribution is generated from a Dirichlet-$1/2$ prior.
The Braess-Sauer estimator has
$\beta_{BS}(0) = 1/2, \beta_{BS}(1)= 1, \beta_{BS}(t) = 3/4 \ \forall \, t > 1$ and is min-max optimal for $\rns{\iidk}$.
We also compare against the best regret by any natural estimator, which simply estimates $p(x)$ by $q(x) = \frac{\cpmx{n(x)}}{\prev{n(x)}}$ (see Lemma~\ref{lem:bestnatural}).

We compare the above five estimators for six distributions with support $k = 10000$ and number of samples $n \leq 50000$. All results are averaged over $200$ trials.

\begin{figure}
\centering     %%% not \center
\subfigure[Uniform]{\label{fig:uniform}\includegraphics[width=80mm]{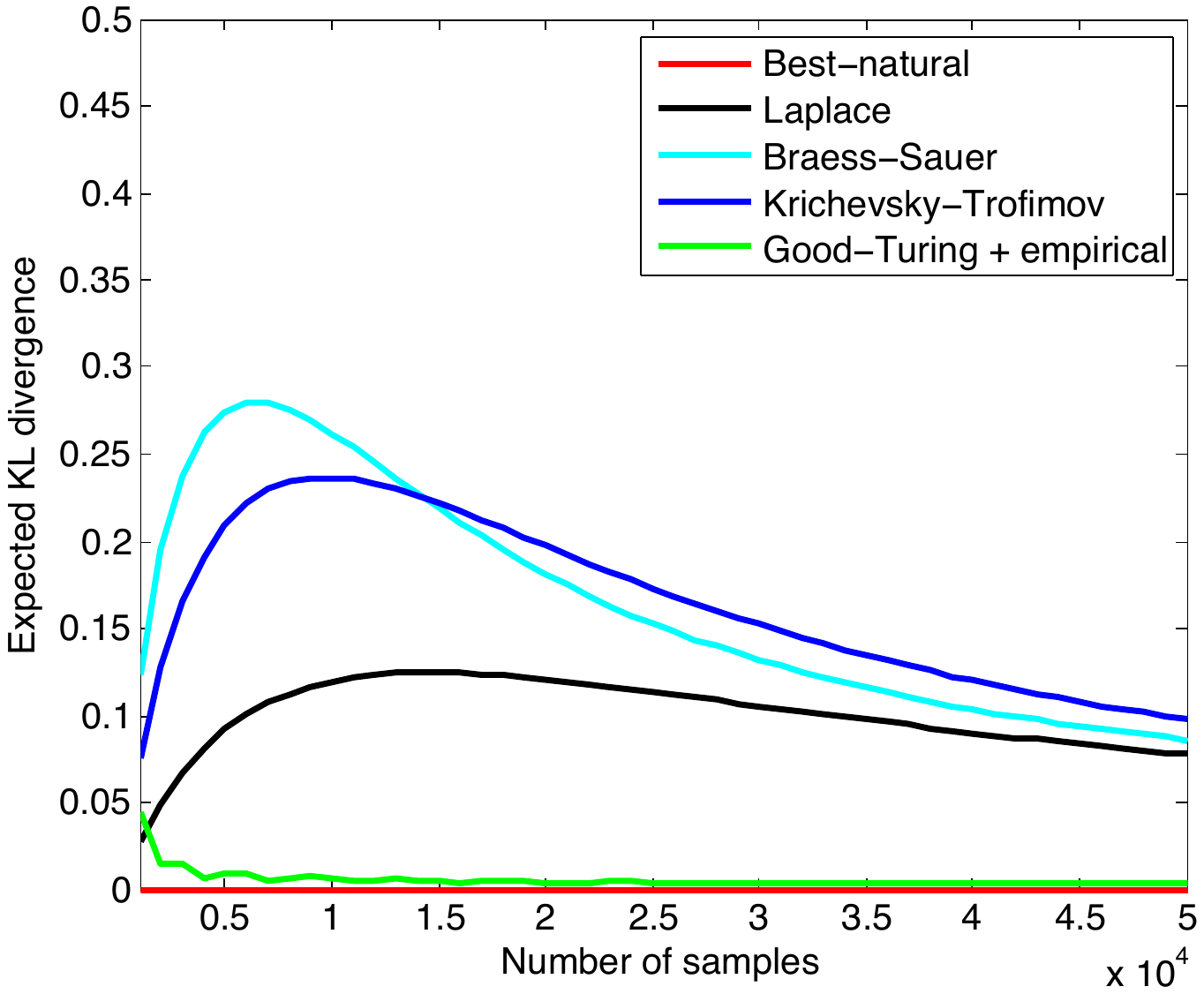}}
\subfigure[Step]{\label{fig:step}\includegraphics[width=80mm]{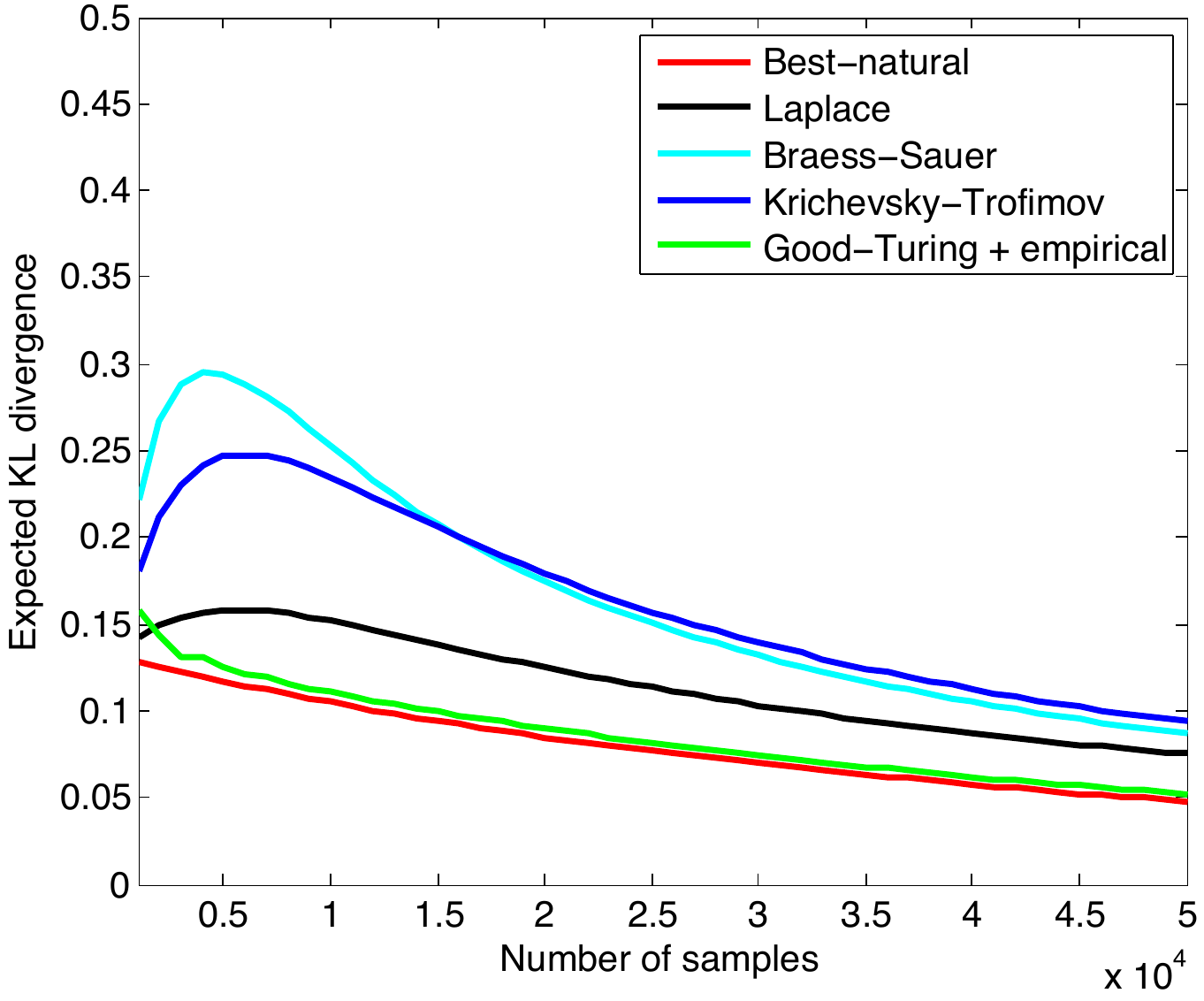}}
\subfigure[Zipf with parameter $1$]{\label{fig:zipf1}\includegraphics[width=80mm]{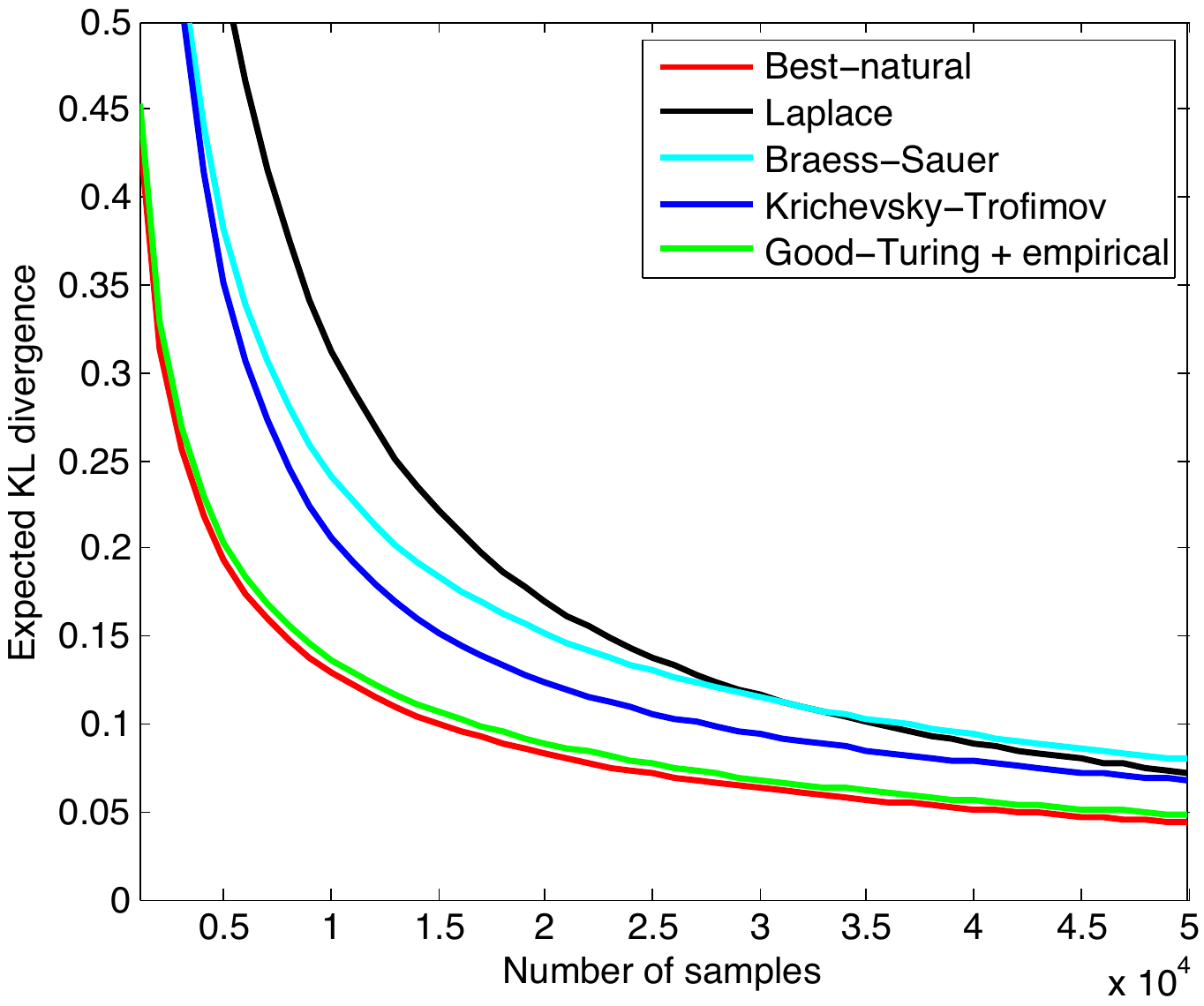}}
\subfigure[Zipf with parameter $1.5$]{\label{fig:zipf15}\includegraphics[width=80mm]{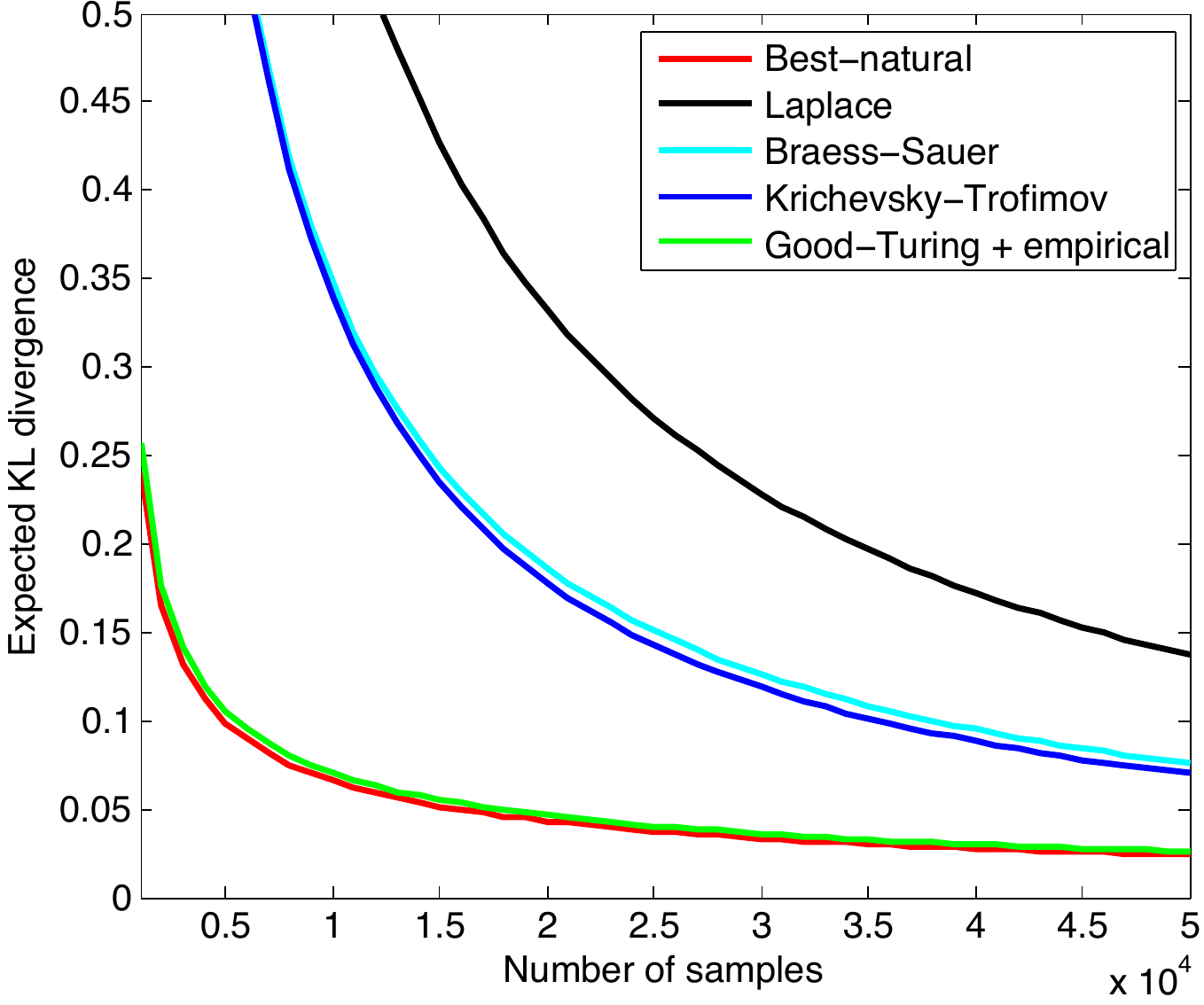}}
\subfigure[Uniform prior (Dirichlet $1$)]{\label{fig:random}\includegraphics[width=80mm]{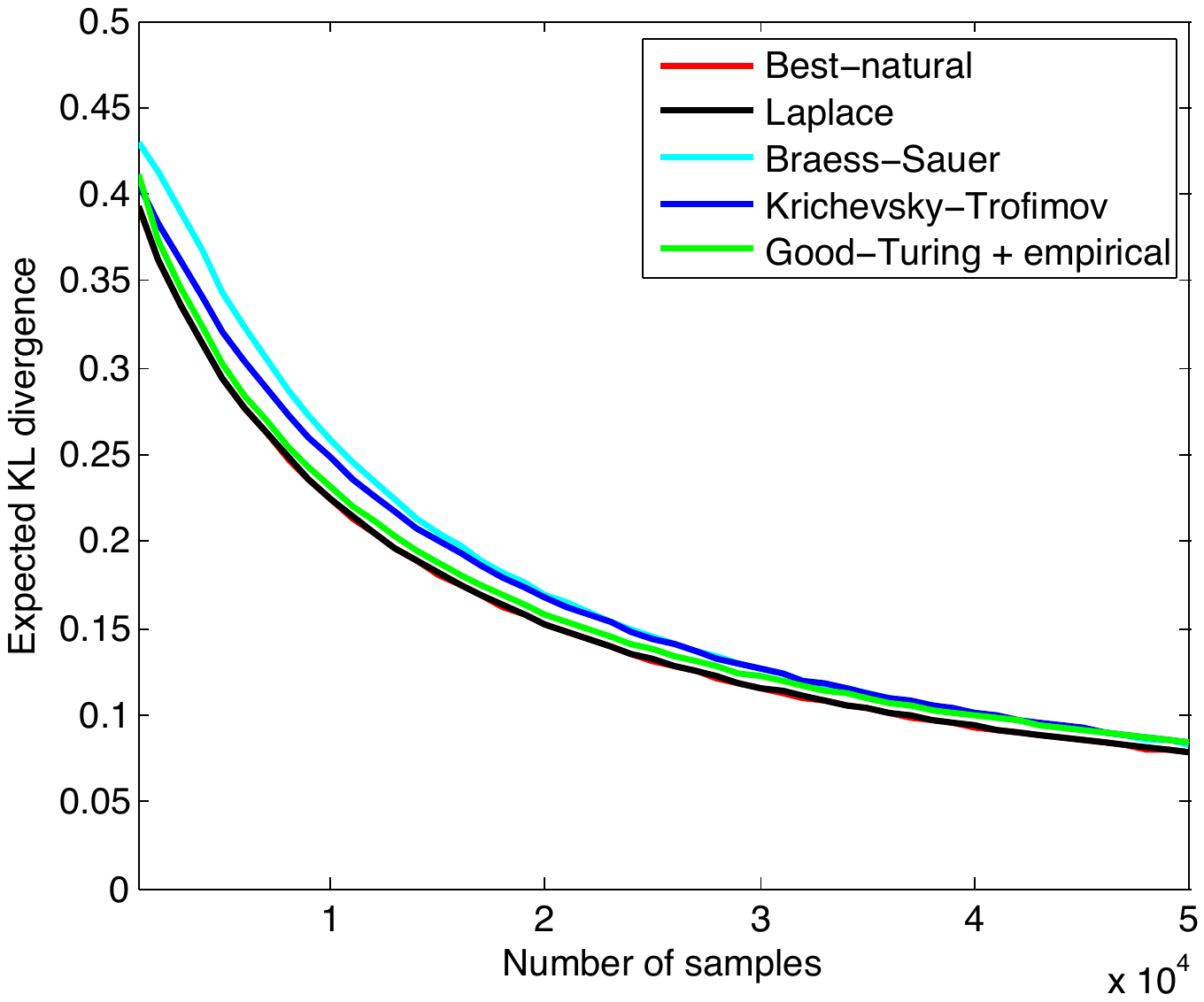}}
\subfigure[Dirichlet $1/2$ prior]{\label{fig:gammahalf}\includegraphics[width=80mm]{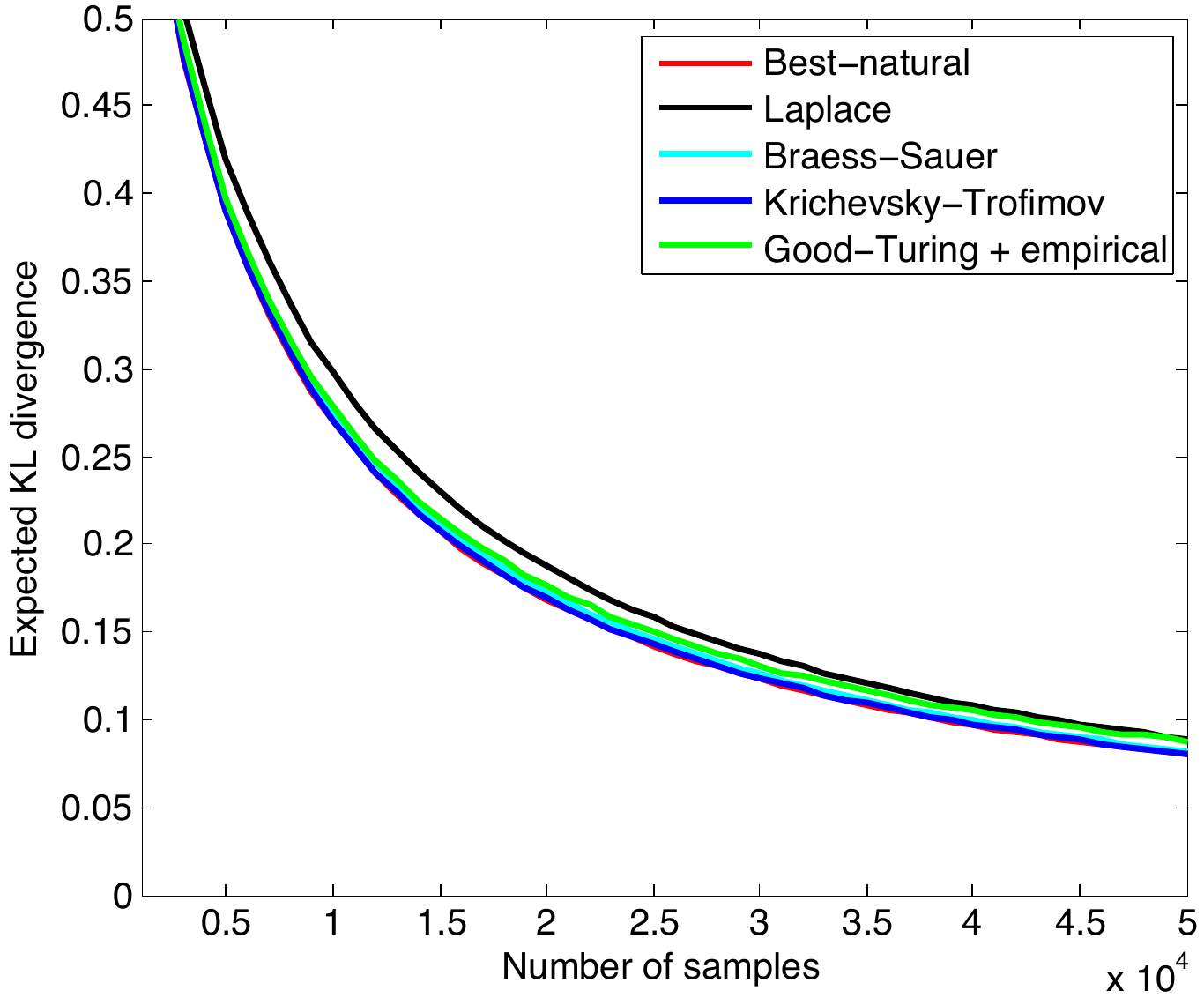}}
\caption{Simulation results for support $10000$, 
number of samples ranging from $1000$ to $50000$, averaged over $200$ trials.}
\label{fig:experiments}
\end{figure}

The six distributions  are  uniform distribution, step distribution with half the symbols having probability $1/2k$ and the other half have probability $3/2k$, Zipf distribution with parameter $1$ ($p(i) \propto i^{-1}$), Zipf distribution with parameter $1.5$ ($p(i) \propto i^{-1.5}$), a distribution generated by the uniform prior on $\iidk$, and a distribution generated from Dirichlet-$1/2$ prior.

The results are given in Figure~\ref{fig:experiments}.
The proposed estimator
uniformly performs well for  all the six distributions and is close to what the best natural estimator can achieve. Furthermore for Zipf, uniform, and step distributions the performance is significantly better.

%For the uniform distribution (Figure~\ref{fig:uniform}), the best natural estimator has $0$ regret as it can associate probability $1/k$ to the entire support. Our proposed estimator hast the smallest regret compared to other add-$\beta$ estimators. 

The performance of other estimators depend on the underlying distribution. For example, since Laplace is the optimal estimator when the underlying distribution is generated from the uniform prior, it performs well in Figure~\ref{fig:random}, however performs poorly on other distributions. 

Furthermore, even though for distributions generated by Dirichlet priors, all the estimators have similar looking regrets  (Figures~\ref{fig:random},~\ref{fig:gammahalf}), the proposed estimator performs better than estimators which are not designed specifically for that prior.

%
%For distributions generated randomly by either uniform prior (Figure~\ref{fig:random}) notice that Laplce performs same as best natural estimator as Laplace is 
%the best estimator for that prior.  Even though all estimators behave similarly, the proposed estimator performs better than KT and Braess-Sauer estimator. The same observation holds for Dirichlet $1/2$ prior (Figure~\ref{fig:gammahalf}), where the proposed esitmator performs better than Laplace estimator.

%Similarly, Krichevsky-Trofimov estimator performs well for distribution generated by on Dirichlet prior with parameter $1/2$.
%Furthermore, since Laplace estimator forces the distribution to be uniform, it performs reasonably well for very small values of $n$ for uniform and step distributions.

\section{Acknowledgements}
Authors thank Jayadev Acharya, Ashkan Jafarpour, Mesrob Ohannessian, and Yihong Wu for helpful comments.
\bibliographystyle{plainnat}
\bibliography{abr,masterref}
\end{document}